\documentclass[conference]{IEEEtran}

\usepackage[utf8]{inputenc}
\usepackage[T1]{fontenc}
\usepackage{url}
\usepackage{ifthen}
\usepackage{cite}
\usepackage[cmex10]{amsmath} % Use the [cmex10] option to ensure complicance

\usepackage{mypackage}
\usepackage{mathrsfs}

\newtheorem{remark}{Remark}
\usepackage{mathtools}
\usepackage{physics}
\usepackage{siunitx}

\theoremstyle{approximation}

\newcommand{\sinr}{\mathrm{SINR}}
\newcommand{\sir}{\mathrm{SIR}}

\newcommand{\Eb}{\mathbb{E}}

\usepackage{physics}

\usepackage[nomain,acronym,shortcuts]{glossaries}
\makeglossaries
\newcommand*{\acro}[3][]{\newacronym[#1]{#2}{#2}{#3}}

\acro{D2D}{device-to-device}
\acro{SIR}{signal to interference ratio}
\acro{SINR}{signal to interference plus noise ratio}
\acro{PCP}{Poisson cluster process}
\acro{CoMP}{coordinated multi-point}
\acro{BS}{base station} 
\acro{MD-CoMP}{macrodiversity CoMP transmission}
\acro{MAC}{medium-access-control}
\acro{JT-CoMP}{joint transmission CoMP}
\acro{CoMP-JT}{coordinated multipoint joint transmission}
\acro{SBS}{small base station}
\acro{MDSD}{multiple devices to the single device}
\acro{MDS}{maximum distance seperable}
\acro{SCN}{small cell network}
\acro{PPP}{Poisson point process}
\acro{TCP}{Thomas cluster process}
\acro{CSI}{channel state information}
\acro{PDF}{probability distribution function}
\acro{PMF}{probability mass function}
\acro{RV}{random variable}
\acro{i.i.d.}{independently and identically distributed}
\acro{MBMS}{multimedia broadcasting multicasting service}
\acro{EE}{energy efficiency}
\acro{HCP}{hard-core placement}
\acro{CCDF}{complementary cumulative distribution function}
\acro{PC}{probabilistic caching}
\acro{RC}{random uniform caching}
\acro{CPF}{caching popular files}
\begin{document}
\title{On Minimizing Energy Consumption for D2D Clustered Caching Networks}

\author{\IEEEauthorblockN{Ramy Amer\IEEEauthorrefmark{1},			
M. Majid Butt\IEEEauthorrefmark{1}\IEEEauthorrefmark{2},
Hesham ElSawy\IEEEauthorrefmark{3},
Mehdi Bennis\IEEEauthorrefmark{4}, %\IEEEauthorrefmark{5}, and
Jacek Kibi\l{}da\IEEEauthorrefmark{1},
Nicola Marchetti\IEEEauthorrefmark{1}~\IEEEmembership{Fellow,~IEEE}}
\IEEEauthorblockA{\IEEEauthorrefmark{1}CONNECT, Trinity College, University of Dublin, Ireland}
\IEEEauthorblockA{\IEEEauthorrefmark{2}School of Engineering, University of Glasgow, UK}
\IEEEauthorblockA{\IEEEauthorrefmark{3}King Abdullah University for Science and Technology (KAUST)}
\IEEEauthorblockA{\IEEEauthorrefmark{4}Centre for Wireless Communications, University of Oulu, Finland}

%\IEEEauthorblockA{\IEEEauthorrefmark{5}Department of Computer Engineering, Kyung Hee University, South Korea}
\IEEEauthorblockA{email:\{ramyr, majid.butt, kibildj, nicola.marchetti\}@tcd.ie, bennis@ee.oulu.fi, hesham.elsawy@kaust.edu.sa}}
%\thanks{Manuscript received December 1, 2012; revised August 26, 2015.
%Corresponding author: M. Shell (email: http://www.michaelshell.org/contact.html).}}

\maketitle

%%%%%%
%% Abstract:
%% If your paper is eligible for the student paper award, please add
%% the comment "THIS PAPER IS ELIGIBLE FOR THE STUDENT PAPER AWARD
%% AWARD." as a first line in the abstract.
%% For the final version of the accepted paper, please do not forget
%% to remove this comment!
%%
\begin{abstract}
We formulate and solve the energy minimization problem for a clustered \ac{D2D} network with cache-enabled mobile devices. Devices are distributed according to a \ac{PCP} and are assumed to have a surplus memory which is exploited to proactively cache files  from a library. Devices can retrieve the requested files from their caches, from neighboring devices in their proximity (cluster), or from the base station as a last resort.
We minimize the energy consumption of the proposed network under a random probabilistic caching scheme, where files are independently cached according to a specific probability distribution. A closed form expression for the \ac{D2D} coverage probability is obtained. The energy consumption problem is then formulated as a function of the caching distribution, and the optimal probabilistic caching distribution is obtained.
Results reveal that the proposed caching distribution reduces energy consumption up to 33\% as compared to caching popular files scheme. 
%Simulation results show a delay reduction of up to 33\% as compared to caching popular files.
%better performance is achieved in terms of energy consumption as compared to other caching schemes proposed earlier in literature.


\end{abstract}
%\keyword{\ac{D2D}}
%I'd spend additional time to clearly state the contribution you make in this paper (second to the last paragraph in your introduction). IMHO, there are three:
%i) We propose a content-distribution network model with spatial clustering among devices belonging to the same cluster
%ii) We show that the probability of coverage in a \ac{D2D} network with clustered devices and round-robin-based scheduling of transmissions is equivalent to the probability of coverage of a \ac{PPP}  network of intensity XXX
%iii) We show that for such a setup optimal system energy consumption can be achieved not by exploiting the popularity of the cached files but by exploiting the spatial distribution of the nodes in the cluster.
%%%%%%%%%%%%%%%%%%%%%%%%%%%%%%%%%%%%%%%%%%%%%%%%%%%%%%%%%%%%%%%
\section{Introduction}
Caching at mobile devices significantly improves system performance by facilitating  \ac{D2D} communications, which enhances the spectral efficiency and alleviates the heavy burden on backhaul links \cite{Femtocaching}. Modeling the cache-enabled heterogeneous networks, including small cell \ac{BS} and mobile devices, follows two main directions in the literature. The first line of work focuses on the fundamental throughput scaling results by assuming a simple protocol channel model \cite{Femtocaching,golrezaei2014base,amer2017delay,8412262}, where two devices can communicate if they are within certain distance. The second line of work, which is relevant to our work, considers a more realistic model for the underlying physical layer \cite{andreev2015analyzing,cache_schedule}. We review some of the works relevant to the second line, focusing mainly on the \ac{EE} of wireless caching networks.

%As an early work that employs stochastic geometry as an analytic tool in the wireless caching problem, the authors in \cite{cache_schedule} maximized the offloading gain of cache-enabled  \ac{D2D} communications wherein the devices are distributed according to Poisson point process (PPP). The authors jointly optimize caching and scheduling policies to maximize successful offloading probability.
\ac{EE} in wireless caching networks is widely studied in the literature.
For example, an optimal caching problem is formulated in \cite{energy_efficiency} to minimize energy consumption in a wireless network. The authors consider a cooperative wireless caching network where relay nodes cooperate with the devices to cache the most popular files in order to minimize energy consumption. In \cite{ee_BS}, the authors quantify the effect of caching on the \ac{EE} of wireless access networks. It is shown that \ac{EE} gain from caching is a function of content popularity, backhaul capacity, and the interference level. However, studies in \cite{energy_efficiency,ee_BS} focus on caching the files at the BSs or access points. 

%TWC
Content placement at the device level in \ac{D2D} clustered networks is a viable approach to improve the network performance \cite{clustered_twc,clustered_tcom,8070464}.
%in cluster-based \ac{D2D} networks is considered in several works, e.g., 
The authors in \cite{clustered_twc} developed a stochastic geometry based model to characterize the performance of cluster-centric content placement in a \ac{D2D} network. % under two strategies. First, when each device randomly chooses its serving device from its local cluster, and secondly, when each device connects to its $k$-th closest transmitting device from its local cluster. 
The authors characterize the optimal number of \ac{D2D} transmitters that must be simultaneously activated in each cluster to maximize the area spectral efficiency. 
The performance of cluster-centric content placement is characterized in \cite{clustered_tcom}, where the content of interest
in each cluster is cached closer to the cluster center, such that the collective performance of all the devices in each cluster is
optimized. Inspired by the Matern hard-core point process that captures pairwise interactions between nodes, the authors in \cite{8070464}  devised a novel spatially correlated caching strategy called \ac{HCP} such that the \ac{D2D} devices caching the same content are never closer to each other than the exclusion radius. However, the problem of \ac{EE} for \ac{D2D} based caching is not yet addressed in the literature.

%clustered_tcom
%In \cite{clustered_tcom}, the performance of the cluster-centric content placement is characterized, where a content of interest in each cluster is cached closer to the cluster center, such that the overall performance of all the devices in each cluster is optimized.

In this paper, we aim at minimizing the energy consumption of a clustered cache-enabled \ac{D2D} network, where devices have unused memory to cache files of interest. Devices are assumed to cache files according to a random probabilistic caching scheme. We consider a \ac{PCP} where the devices are spatially distributed as groups in clusters. The cluster centers are drawn from a \ac{PPP}, and the cluster members are normally distributed around the centers. We derive the \ac{D2D} coverage probability and calculate the average achievable rate. The energy consumption problem is then formulated and shown to be convex, and the optimal caching distribution is obtained.
Results unveil that the obtained optimal probabilistic caching scheme yields the lowest energy consumption among all conventional benchmark caching schemes.
%We validate our theoretical findings via simulations. Results unveil that the obtained optimal upper bound on the energy consumption yields the best performance.

%%%%%%%%%%%%%%%%%%%%%%%%%%%%%%%%%%%%%%%%%%%%%%%%%%%%%%%%%%%%%%
The rest of this paper is organized as follows. Section II and Section III present the system model and energy consumption problem formulation, respectively. The coverage probability analysis is discussed in Section IV and the energy minimization problem is solved in Section V. Numerical results are then presented in Section VI, and Section VII concludes the paper. 
%%%%%%%%%%%%%%%%%%%%%%%%%%%%%%%%%%%%%%%%%%%%%%%%%%%%%%%%%%%%%%%

\section{System Model}

\subsection{System Setup}
We model the location of the mobile devices with a \ac{PCP}  in which the parent points are drawn from a \ac{PPP}  $\Phi_p$ with density $\lambda_p$, and the daughter points are normally scattered with variance $\sigma^2 \in \mathbb{R}^2$ around each parent
point \cite{haenggi2012stochastic}. %independently and identically distributed (i.i.d.)
 The parent points and offspring are referred to as cluster centers and cluster members, respectively. The number of cluster members in each cluster is considered constant, and denoted as $n$. It is worth highlighting that the well-known \ac{TCP}, see \cite{clustered_info}, is similar to the considered \ac{PCP} but the number of cluster members follows a Poisson distribution.  
Therefore, our setup can be interpreted as a variant of \ac{TCP}. Note that the choice of such a variant of \ac{TCP} is in fact for ease of analysis. The density function of the location of a cluster member relative to its cluster center, $y \in \mathbb{R}^2$, is
\begin{equation}
f_Y(y) = \frac{1}{2\pi\sigma^2}\textrm{exp}\Big(-\frac{\lVert y\rVert^2}{2\sigma^2}\Big),
\label{pcp}
\end{equation}
where $\lVert .\rVert$ is the Euclidean norm. The intensity function of a cluster is given by $\lambda_c(y) = \frac{n}{2\pi\sigma^2}\textrm{exp}\big(-\frac{\lVert y\rVert^2}{2\sigma^2}\big)$. Therefore, the overall intensity of the process is given by $\lambda = n\lambda_p$. %The depicted \ac{D2D} network model is illustrated in Fig. \ref{Network Model}.
We assume that the BSs' distribution follows another \ac{PPP}   $\Phi_{bs}$ with density $\lambda_{bs}$, which is independent of $\Phi_p$. 
%%\vspace{-0.2cm}
%\begin{figure}[t] %%[htbp]
%\centering
%\includegraphics[width=0.3\textwidth]{Figures/ch3/point_process2}
%%\vspace{-0.4cm}			
%\caption {Schematic diagram of the proposed \ac{PCP}.} %with $\sigma=0.2$ and $\lambda_p = 0.5$.}
%%$\lambda_{bs} = ?$.}
%\label{Network Model}
%%\vspace{-0.5cm}
%\end{figure}

%If the number of cluster members in each cluster is Poisson distributed, this setup corresponds to the well-known Thomas
%cluster process \cite{clustered_info}. Therefore, our setup can be interpreted as a variant of \ac{TCP} .
%However, to simplify the analysis in the sequel, we assume that the total number of devices per cluster is fixed and equal to $n$.

\subsection{Content Popularity and Probabilistic Caching Placement}
We assume that each device has a surplus memory of size $M$ designated for caching files. 
%all the devices are equipped with an internal memory with a cache  size of $M$ files, i.e., $M$ different files can be cached at each device. 
The total number of files is $N_f> M$ and the set (library) of content indices is denoted as $\mathcal{F} = \{1, 2, \dots , N_f\}$. These files represent the content catalog that all the devices in a cell may request, which are indexed in a descending order of popularity. The probability that the $i-th$ file is requested follows a Zipf's distribution \cite{breslau1999web} given by,
\begin{equation}
q_i = \frac{ i^{-\beta} }{\sum_{k=1}^{N_f}k^{-\beta}},
\label{zipf}
\end{equation}
where $\beta$ is a parameter that reflects how skewed the popularity distribution is. For example, if $\beta= 0$, the popularity of the files has a uniform distribution. Increasing $\beta$ increases the disparity among the files popularity such that lower indexed files have higher popularity. By definition, $\sum_{i=1}^{N_f}q_i = 1$. 
It is assumed that different clusters may have different interests, and hence, different library content and/or popularity distribution.  
 We use Zipf's distribution to model the popularity of files per cluster.

%We consider a clustered \ac{D2D} network where some of the content of interest for devices of a given cluster is cached at the mobile devices.
%We adopt a random content placement where each device independently selects a file to cache according to a specific probability distribution $\textbf{c} = \{b_1, b_2, \dots, b_{N_{f}}\}$, where $b_i$ is the probability that a device caches the $i-th$ file, $0 \leq b_i \leq 1$ for all $i \in \mathcal{F}$. According to the caching probabilities (policies), each device randomly builds a list of $M$ files to cache based on the probabilistic caching method proposed in \cite{geographib_caching}.
%Fig. \ref{prob_cache_example} illustrates an example of the probabilistic caching scheme.
\ac{D2D} communication is enabled within each cluster to deliver popular content. A probabilistic caching model is considered, where the content is randomly and independently placed in the cache memories of different devices in the same cluster, according to the same distribution. The probability that a generic device stores a particular file $i$ is denoted as $b_i$, $0 \leq b_i \leq 1$ for all $i \in \mathcal{F}$.
%This distribution, which is the main design variable in this paper, is given by $\textbf{b} = \{b_1, b_2, \dots, b_{N_{f}}\}$, where $b_i$ is the probability that a device caches the $i-th$ file, $0 \leq b_i \leq 1$ for all $i \in \mathcal{F}$. 
To avoid duplicate caching of the same content within the memory of the same device, we follow the caching approach proposed in \cite{blaszczyszyn2015optimal} as illustrated in Fig. \ref{prob_cache_example}.		%geographib_caching

%%\vspace{-0.2cm}
\begin{figure}[t]
\centering
\includegraphics[width=0.3\textwidth]{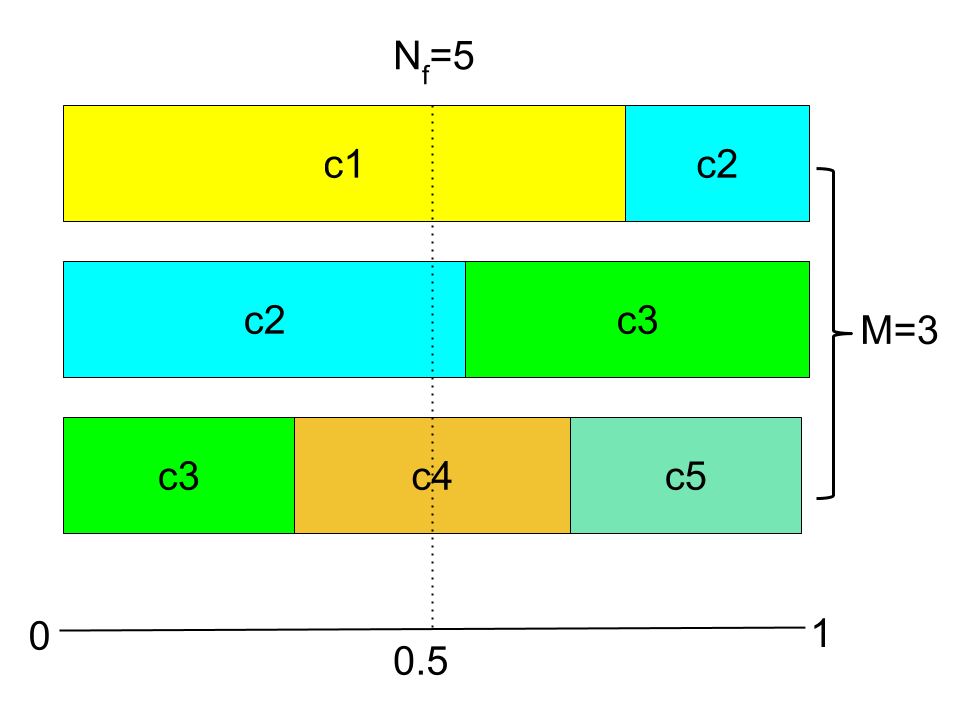}
%\vspace{-0.3cm}				
\caption {The cache memory of size $M = 3$ is equally divided into $3$ blocks of unit size. Starting from content $i=1$ to $i=N_f$, each content sequentially fills these $3$ memory blocks by an amount $b_i$. The amounts (probabilities) $b_i$ eventually fill all $3$ blocks since $\sum_{i=1}^{N_f} b_i = M$ \cite{blaszczyszyn2015optimal}. Then a random number $\in [0,1]$ is generated, and a content $i$ is chosen from each block, whose $b_i$ fills the part intersecting with the generated random number. In this way, in the given example, the contents $\{1, 2, 4\}$ are chosen to be cached.
}
\label{prob_cache_example}
%\vspace{-0.5cm}
\end{figure}

If a device caches the desired file in its own memory, the device directly retrieves the content. However, if the device does not cache the desired file, it downloads the requested file from one neighboring device in the same cluster (that caches the desired file) via \ac{D2D} communication. Otherwise, the device attaches to the nearest BS as a last resort to download the content which is not cached entirely within the device's cluster. We assume that the \ac{D2D} communication is operating as out-of-band D2D. Also, it is assumed that all the \ac{D2D} transmitters transmit with the same power $P_d$, and all BSs transmit with the same power $P_b$. We assume that devices' requests in each cluster are served in a round-robin manner where only one active \ac{D2D} link can be scheduled within each cluster at a time. %For the \ac{D2D} communication, the interference is considered only from the scheduled \ac{D2D} links in the same cell.

\section{PROBLEM FORMULATION} %AND ANALYSIS}

In this section, we formulate the energy consumption minimization problem for the clustered \ac{D2D} caching network. %Then, we
%introduce a definition for the coverage probability which is used for the achievable rate calculation.
%\textbf{Taking the fairness management of popular contents into account, suboptimal caching schemes are then proposed}.
We define the cost $c_{d_i}$ as the time required to download a file $i$ from a neighboring device in the same cluster. Considering the file size $S_i$ of the $i$-th ranked content, the time $c_{d_i}$ is computed as $c_{d_i} = S_i/\overline{R}_d$, where $\overline{R}_d$ denotes the average rate of intra-cluster \ac{D2D} communication. Similarly, we have the time $c_{b_i} = S_i/\overline{R}_b$ when the $i$-th content is served by the BS with average rate $\overline{R}_b$. We assume that the library files have a mean size $\overline{S}$ Mbits.
%We assume that devices' requests are served from the device own cache, a neighboring device in the same cluster, or the BS in order.
Hence, the total consumed energy for all files downloaded within a given cluster can be written similar to \cite{energy_efficiency} as,
%Finally, we assume that each mobile device consumes a relatively small amount of energy to allocate and process the self-cached (own-cached) content for the device; this is represented by the constant value $E^O$.
%time cost $b_O=S_i/R_O$ The average of the total energy consumption $W$ can then be
\begin{equation}
E = \sum_{j=1}^{n} \sum_{i=1}^{N_f}q_i \mathcal{P}_{i,j}^d P_d c_{d_i} + q_i \mathcal{P}_{i,j}^b P_b c_{b_i},%+ P_{i}^O q_i E^O\},
\label{energy}
\end{equation}
where $\mathcal{P}_{i,j}^d$ and $\mathcal{P}_{i,j}^b$ represent the probability of obtaining the $i$-th file by the $j$-th device from the local cluster, i.e., via \ac{D2D} communication, and the BS, respectively. Since we assume that the caching distribution is the same for all devices, the index $j$ is henceforth dropped. %%
For the self-caching probability, it is clear that $\mathcal{P}_{i}^s=b_i$ since a device caches the $i$-th file with probability $b_i$.
 For the probability $\mathcal{P}_{i,j}^d$, we define two events:\\
 %of obtaining $i$-th file from a neighboring device via \ac{D2D} communication
\textbf{Event A:} The requesting device does not cache file $i$.
\\
\textbf{Event B:} The requested file is cached by any other device within the same cluster.
Then, the probability $\mathcal{P}_{i}^d$ represents the probability of intersection of these two events, written as
\begin{align}
\mathcal{P}_{i}^d =\mathbb{P}\big(A \cap B\big),
\label{zipf}
\end{align}
Since the two events $A$ and $B$ are independent, then 
\begin{align}
\mathcal{P}_{i}^d&=\big(1 - b_i\big)\big(1-(1-b_i)^{n-1}\big), \nonumber \\
&= \big(1 - b_i\big) -(1-b_i)^{n},
\label{d2d_prob}
 \end{align}
 where $(1 - b_i)$ is $\mathbb{P}(A)$, and $\big(1-(1-b_i)^{n-1}\big)$ is $\mathbb{P}(B)$. Therefore, we have,
 \begin{align}
\mathcal{P}_{i}^b &= (1-b_i)^{n}
%1 - \big(b_i +  \big(1 - b_i\big) -(1-b_i)^{n} \big),\nonumber \\ &=(1-b_i)^{n} 
\label{BS_prob}
\end{align}
In the following, the achievable rate analysis for the \ac{D2D} and BS communication is presented. These rates are then used to derive a  closed-form expression for the energy consumption in Section V.
%we continue calculate the achievable rate for the \ac{D2D} and BS communication consequently to get a closed form of the energy consumption.

\section{Average Achievable Rate}

%We strive to obtain a closed form expression for the average transmission rate on the \ac{D2D} and the cellular links at a given cluster. 
A fixed rate transmission model is adopted in our study, where each TX (\ac{D2D} or BS) transmits at a fixed rate log$_2(1+\theta)$ bits/s/Hz, where $\theta$ is a design parameter. Since, the rate is fixed, the transmission is subject to outage due to fading and interference fluctuations. Consequently, the de facto transmissions rate (i.e., throughput) is given by 
%%%We adopt a transmission scheme with constant parameter that has SINR threshold $\theta$ \cite{7733098}.
\begin{equation}
\label{rate_eqn}
\overline{R}_m = \textrm{W$_{m}$log$_{2}$}(1+ \theta) p_{c},
\end{equation}
where W$_{m}$ is the allocated bandwidth for the communication link $m \in \{{\rm D2D},{\rm BS}\}$, $\theta$ is a pre-determined threshold for successful reception, and $p_c$ is the coverage probability. $p_c$ is defined as
the probability that the \ac{SINR} of the link of interest at the receiver exceeds the required threshold for successful demodulation and decoding, i.e.,
\begin{equation}
p_c = \mathbb{E}(\textbf{1} \{\sinr>\theta\}),
\end{equation}
where $\textbf{1}\{.\}$ is the indicator function.
%\begin{definition}{Average Achievable Rate \cite{clustered_tcom}}:				%%%%%
%The average number of bits transmitted per unit time can be defined as:
%\end{definition}

\subsection{Average Achievable Rate for \ac{D2D} communication}
Following the methodology in \cite{andrews2011tractable}, we compute the \ac{D2D} coverage probability $p_{c_d}$.
Without loss of generality, we conduct the analysis for a cluster whose center is located at $x_0\in \Phi_p$ (referred to as representative cluster), and the device that requests the content (henceforth called typical device) is located at the origin.
We assume that the D2D-TX  that caches the requested file is located at $y_0$ w.r.t. $x_0$, where $x_0, y_0\in \mathbb{R}^2$. The distance from the D2D-TX to the typical device (D2D-RX of interest) is  denoted as $r=\lVert x_0+y_0\rVert$, which is a realization of a random variable $R$ whose distribution is described later. This setup is illustrated in Fig. \ref{distance}.
The received power at the D2D-RX of interest is expressed as
\begin{align}
P &= P_d  g_0 \lVert x_0+y_0\rVert^{-\alpha}= P_d  g_0 r^{-\alpha}			
%P &= P_d  g_0 r^{-\alpha},
\label{pwr}
\end{align}
where $g_0 \sim $ exp(1) is an exponential random variable which models Rayleigh fading and $\alpha > 2$ is the path loss exponent. 

It is assumed that there is always one active \ac{D2D} link per cluster. 
%That is the probability that there is no opportunity to transfer a file between two devices is negligible.
%\textbf{For the received interference power from the inter-cluster transmission, it is assumed that a cluster is active if and only if a requested file is served via \ac{D2D} communication, i.e., any of the remote clusters is active with the same probability  $p =\mathcal{P}_{i}^d$, defined in (\ref{d2d_prob}).} 
Similar to (\ref{pwr}), the interference from the simultaneously active D2D-TXs outside the representative cluster $x_0 \in \Phi_p$  at the typical device is expressed as
\begin{align}
I_{\Phi_p^{!}} &= \sum_{x\in \Phi_p \setminus\{x_0\}}  P_d g_{yx}  \lVert x+y\rVert^{-\alpha}\\
& =  \sum_{x\in \Phi_p \setminus\{x_0\}}  P_d g_{u}  u^{-\alpha}
\end{align}
where $y$ is the marginal distance from the potential interfering device in a given cluster to its cluster center at $x \in \Phi_p$.
$u = \lVert x+y\rVert$ is a realization of a random variable $U$ modeling the interfering distance (shown in Fig. \ref{distance}),
$g_{yx} \sim $ exp(1) are independently and identically distributed exponential random variables modeling Rayleigh fading, and $g_{u} = g_{yx}$ for ease of notation.
%shown in Fig. \ref{distance}
%To better understand the notation of the distances, Fig. (\ref{distance}) shows an illustration of the above notations, where $x_0$ and $x$ are cluster centers, $y_0$ and $y$ are the marginal distance from the transmitting devices to the corresponding cluster center, and $r$ and $u$ are the serving and interfering distances, respectively.
%%\vspace{-0.4cm}

\begin{figure}[t]
\centering
\includegraphics[width=0.4\textwidth]{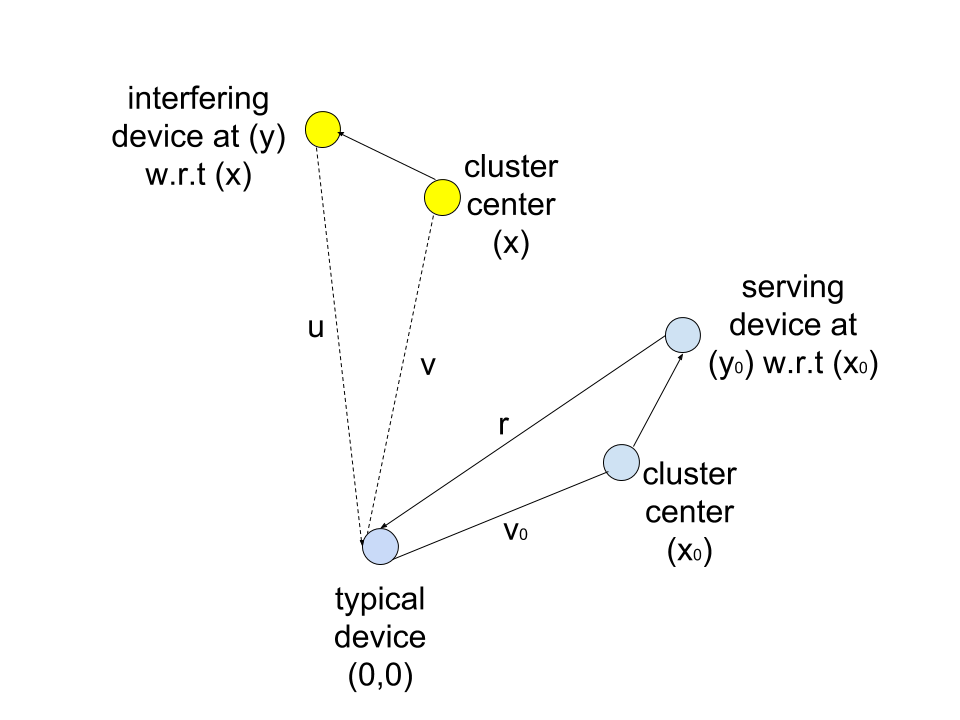}	
%\vspace{-0.4cm}			
\caption {Illustration of the representative cluster and one interfering cluster.}
\label{distance}
%\vspace{-0.5cm}
\end{figure}

Assuming that the thermal noise is neglected compared to inter-cluster interference, the \ac{SIR} at the typical device is
\begin{equation}
\sir =  \frac{P }{I_{\Phi_p^{!}}} = \frac{P_d  g_0 r^{-\alpha}}{I_{\Phi_p^{!}} } %\sum_{x\in \Phi_p \setminus\{x_0\}}  P_d g_{u}  u^{-\alpha}
\frac{}{}
\end{equation}
For the \ac{D2D} coverage probability, we proceed as follows,
%\textbf{UP to now}, we need to further simplify the above equation of the Laplace transform of the interference, and then it is gonna be used coincide with the distance distributions to obtain the coverage probability from this equation:
\begin{align}
 p_{c_d} &= \mathbb{P} [\sir > \theta] = \mathbb{E}_R  \mathbb{P}\Big[\frac{g_0 P_d r^{-\alpha}}{I_{\Phi_p^{!}} }>\theta| R=r\Big]	 \nonumber \\
 %\int_{r>0}\mathbb{P}\big[\sir>\theta \big] f_{R}(r){\rm dr} \\
 &= \int_{r>0}\mathbb{P}\Big[\frac{g_0 P_d r^{-\alpha}}{I_{\Phi_p^{!}} }>\theta\Big] f_{R}(r){\rm dr}  \nonumber \\
 &= \int_{r>0}\mathbb{P}\Big[g_0> \theta r^{\alpha}{\frac{I_{\Phi_p^{!}}}{P_d} }\Big] f_{R}(r){\rm dr},
 \end{align}
 where $f_{R}(r)$ is the serving distance distribution. Since $g_0 \sim $ exp (1), we have
 \begin{align}
 p_{c_d} &= \int_{r>0}\mathbb{E}\Big[\text{exp}\big(\frac{-\theta r^{\alpha}}{P_d}{I_{\Phi_p^{!}} }\big)\Big]f_{R}(r){\rm dr}, \nonumber  \\
 \label{cov_eqn}
&\overset{(a)}{=} \int_{r>0}\mathscr{L}_{I_{\Phi_p^{!}}}\big(s=\frac{\theta r^{\alpha}}{P_d}\big)f_{R}(r){\rm dr},
 \end{align}
where (a) follows from the Laplace transform definition.
Now, we proceed with the derivation of the Laplace transform of the inter-cluster interference to obtain the \ac{D2D} coverage probability.
\begin{lemma}
\label{Ricepdf_lemma}
 Laplace transform of the inter-cluster interference can be expressed as
\begin{align}
\label{laplace_trans}
\mathscr{L}_{I_{\Phi_p^{!}}} (s)&= {\rm exp}\Big(-\pi\lambda_p (sP_d )^{2/\alpha} \Gamma(1 + 2/\alpha)\Gamma(1 - 2/\alpha)\Big)
\end{align}
 %conditioned on the distance from the typical device to the center of the interfering cluster being $v$}.%\footnote{Rice distribution: $\mathrm{Rice} (u|v,\sigma) = \frac{u}{\sigma^2}\mathrm{exp}\big(\frac{-u^2 + v^2}{2\sigma^2}\big) I_0\big(\frac{uv}{\sigma^2}\big)$, where $I_0$ is the modified Bessel function of the first kind with order zero.}		
% $v$ is a random variable representing the distance from the typical device to the center of the interfering cluster.
\end{lemma}
\begin{proof}
Please see Appendix A.
\end{proof}
\begin{lemma}
\label{coverage_lemma}
 Substituting the obtained Laplace transform of the inter-cluster interference (\ref{laplace_trans}) into (\ref{cov_eqn}) yields the \ac{D2D} coverage probability
 \begin{align} 
 \label{v_0}
 p_{c_d} = \frac{1}{4\sigma^2 \mathcal{Z}(\theta,\alpha,\sigma)} , 
 \end{align}
 {\rm where $\mathcal{Z}(\theta,\alpha,\sigma)= (\pi\lambda_p \theta^{2/\alpha}\Gamma(1 + 2/\alpha)\Gamma(1 - 2/\alpha)+ \frac{1}{4\sigma^2})$}
 %{\rm where $f_R(r)=\mathrm{Rayleigh}(r,\sqrt{2}\sigma)$ represents Rayleigh's probability density function of scale parameter $\sqrt{2}\sigma$.}
\end{lemma}			%\mathrm{}
\begin{proof}
Please see Appendix B.
%\footnote{Please refer to the arxiv version of this paper for full technical details \cite[Lemma 2]{full_text}.}
%Averaging over the serving distance using the distribution of $f_R(r)$ yields the coverage probability in (\ref{v_0}).
\end{proof}
\begin{remark} {\rm The Laplace transform of interference of the \ac{PCP} with one active \ac{D2D} link per cluster in (\ref{laplace_trans}) is similar to that of the \ac{PPP}   \cite{andrews2011tractable}. This can be explained with the displacement theory of the \ac{PPP}  \cite{daley2007introduction}, as each interferer is a point of a \ac{PPP}  that is displaced randomly and independently of all other points.}
%This result is coincident with the conclusion for the coverage probability in the networks modeled by the PPPs [40].
%The form of the conditional pdf given by (21) is very similar to the contact distance distribution of the \ac{PPP}  given by (6). The only difference lies in that the intensity of the \ac{PPP}  xx is replaced by xx Pk
\end{remark} 
\begin{remark} {\rm 
The obtained \ac{D2D} coverage probability in (\ref{v_0}) for the \ac{PCP} with one active \ac{D2D} link per cluster shows two important insights when compared to an equivalent PPP. Firstly, it depends on the density of the parent \ac{PPP}  $\lambda_p$ which directly affects the interfering distances. However, in the PPP, the effect of device density does not exist since it equally affects both the serving and interfering distances when a device associates to its nearest serving device. Secondly, it  depends on the serving distance which is represented by the displacement standard deviation $\sigma$.}% from the cluster center.}
\end{remark}
%as in \cite{andrews2011tractable}

\subsection{Average Achievable Rate for Cellular Communication}
For the assumed \ac{PPP}  model for the BSs' distribution, and based on the nearest BS association principle, it is shown in \cite{ganti2016asymptotics} that the BS coverage probability can be expressed as
\begin{equation}
p_{c_b} =\frac{1}{{}_2 F_1(1,-\delta;1-\delta;-\theta)},
%1+\sqrt{\theta}(\frac{\pi}{2} - \textrm{arctan}(\frac{1}{\sqrt{\theta}}))},
\label{p_b_bs}
\end{equation}
where ${}_2 F_1(.)$ is the Gaussian hypergeometric function and $\delta = 2/\alpha$. By substituting the obtained coverage probabilities from (\ref{v_0}) and (\ref{p_b_bs}) into (\ref{rate_eqn}), $\overline{R}_d$ and $\overline{R}_b$ are derived.

\section{Energy Consumption Minimization}
The energy minimization problem is formulated as
\begin{align}
\label{optimize_eqn1}
&\underset{b_i}{\text{min}} \quad E = n\sum_{i=1}^{N_f} q_i\mathcal{P}_{i}^d P_d \frac{S_i}{\overline{R}_d} + q_i \mathcal{P}_{i}^b P_b \frac{S_i}{\overline{R}_b} \\
\label{const11}
&\textrm{s.t.}\quad  \sum_{i=1}^{N_f} b_i = M, \quad\quad b_i \in [ 0, 1],
\end{align}	
where (\ref{const11}) is the device cache size constraint to avoid duplicate caching of the same file (cf. Fig. \ref{prob_cache_example}). It is easy to show the convexity of the objective function $E$ by confirming that the Hessian matrix of $E$ w.r.t. the caching variables, \textbf{H}$_{i,j} = \frac{\partial^2 E}{\partial b_i \partial b_j}$, $\forall i,j \in \mathcal{F}$, is positive semidefinite (and hence $E$ is convex).
Utilizing the convexity of $E$, the optimal caching distribution is computed as follows.
%we delegate this proof to the conference version of this paper [1]
%Sketch of proof: Due to lack of space, the proof will be presented in the journal version of this work.
%For brevity, we omit the details of finding the optimal file
%placement ???
%?? . Thus, ???
%?? is given by
\begin{lemma}
The optimal caching distribution $b_i^{*}$ is given by
\begin{align}
b_i^* = \Bigg[ 1 - \Big(\frac{v^* + k^2q_iS_i\frac{P_d}{\overline{R}_d }}{kq_iS_i\big(\frac{P_d}{\overline{R}_d }-\frac{P_b}{\overline{R}_b }\big)} \Big)^{\frac{1}{n-1}}\Bigg]^{+}
\end{align} 
 where $v^{*}$ satisfies the maximum cache constraint $\sum_{i=1}^{N_f} \Big[ 1 - \Big(\frac{v^* + k^2q_iS_i\frac{P_d}{\overline{R}_d }}{kq_iS_i\big(\frac{P_d}{\overline{R}_d }-\frac{P_b}{\overline{R}_b }\big)} \Big)^{\frac{1}{n-1}}\Big]^{+}=M$, and $[x]^+ =$ max$(x,0)$.
\end{lemma}
\begin{proof}
Due to lack of space, we delegate this proof to the journal version of this paper \cite{joint-cache-comm}.
\end{proof}
%%%%%%%%%%%%%%%%%%%%%%%%%%%%%%%%%%%%%%%%%%%%%%%%%%%%%%%%%%%%%%%
\section{Numerical Results}

\begin{table}[ht]
\caption{Simulation Parameters} % title of Table
\centering % used for centering table
\begin{tabular}{c c  c} % centered columns (3 columns)
\hline\hline %inserts double horizontal lines
Description & Parameter & Value  \\ [0.5ex] % inserts table
%heading
\hline % inserts single horizontal line
\ac{BS}-to-Device bandwidth & W$_{bs}$ & \SI{20}{\mega\hertz}  \\ % inserting body of the table
\ac{D2D} bandwidth & W$_{d2d}$ & \SI{20}{\mega\hertz}  \\
\ac{BS} transmission power & $P_b$ & \SI{43}{\deci\bel\of{m}}  \\
\ac{D2D} transmission power & $P_d$ & \SI{23}{\deci\bel\of{m}}  \\
Displacement standard deviation & $\sigma$ & \SI{10}{\metre} \\ %[1ex] % [1ex] adds vertical space
Popularity index&$\beta$&1\\
Path loss exponent&$\alpha$&4\\
Library size&$N_f$&500 files\\
Cache size per device&$M$&10 files\\
Devices per cluster&$n$&10\\
Density of $\Phi_p$&$\lambda_{p}$&50 clusters/\SI{}{km}$^2$ \\
Average content size&$\overline{S}$&\SI{100}{Mbits} \\		%\SI{100}{MBps}	% Mbits
$\sir$ threshold&$\theta$&\SI{0}{\deci\bel}\\
\hline %inserts single line
\end{tabular}
\label{table:sim-parameter} % is used to refer this table in the text
\end{table}

At first, we validate the developed mathematical model via Monte Carlo simulations. Then we benchmark the proposed caching scheme against conventional caching schemes. Unless otherwise stated, the network parameters are selected as shown in Table \ref{table:sim-parameter}. 
%W$_{bs}=$ W$_{d2d}=\SI{20}{\mega\hertz}$, $P_b = \SI{43}{\deci\bel\of{m}}$, $P_d = \SI{23}{\deci\bel\of{m}}$, $\sigma=\SI{10}{\metre}$, $\beta=1, \alpha=4$, $N_f=500$ files, $M=10$ files, $n=10$ devices, $\lambda_{p} =50$ clusters/km$^2$, $\overline{S}=3$ Mbits, and $\theta=\SI{0}{\deci\bel}$. %This simulation setup will be used unless otherwise specified.
%\SIrange{150}{8000}{\kilo\hertz}		
% \SI{150}{\kilo\hertz} to \SI{8}{\mega\hertz} 
%$\lambda_{p} =0.001$ devices/m$^2$ (or 100 devices/km$^2$)
In Fig.~\ref{verify_energy}, we verify the accuracy of the analytical results for the \ac{D2D} coverage probability. 
The close matching between the analytical and simulated results validates the developed mathematical model. 
We see that the coverage probability monotonically decreases with $\sigma$. As $\sigma$ increases, the serving distance increases and the distance between the interferers and the typical device decreases, and equivalently, the \ac{CCDF} of \ac{SIR} decreases. Similarly, it is observed that the \ac{D2D} coverage probability decreases with $\theta$ owing to the decreasing \ac{CCDF} of \ac{SIR}.

%SINR threshold $\theta$.
Fig.~\ref{energy_vs_beta} shows the normalized energy consumption per device versus $\beta$ under different caching schemes, namely, proposed \ac{PC}, \ac{RC}, \ac{CPF}, and Zipf's caching (Zipf). We can see that the minimized  consumed energy under \ac{PC} scheme attains the best performance as compared to other schemes. Also, it is clear that, except for the \ac{RC}, the consumed energy decreases with $\beta$. 
This can be justified by the fact that as $\beta$ increases, fewer files are frequently requested which are more likely to be cached under \ac{PC}, \ac{CPF}, and the Zipf's caching schemes. In the \ac{RC} scheme, files are uniformly chosen for caching independently of their popularity. 
%For the PC, Zipf's and random strategies, it is clear that the energy consumption slightly increases with $\sigma$. As in Fig.~\ref{verify_energy}, as $\sigma$ increases, the \ac{D2D} transmission rate decreases, and correspondingly, the file download time via \ac{D2D} communication increases. This results in higher energy consumption when downloading files via \ac{D2D} communication. However, for the CPF, the energy remains constant since there is no \ac{D2D} communication between devices to  transfer files.

%SINR threshold $\theta$.
We plot the average device energy consumption per device versus number of devices per cluster in Fig.~\ref{energy_vs_n}. First, we see that the normalized energy consumption decreases with the number of devices. As the number of devices per cluster increases, it is more probable to obtain requested files via low power \ac{D2D} communication. When the number of devices per cluster is relatively large, the normalized energy consumption tends to flatten as most of the content becomes cached at the cluster devices. 
%^^^^^^^^^^We also see that, for all caching schemes, the energy consumption increases with $\theta$. This is because as $\theta$ increases, the probability of successful transmission decreases, and accordingly, more energy is consumed for retransmission. ^^^^^
%%%the obtained optimal upper bound on the consumed energy under probabilistic caching scheme attains the best performance among other schemes. Also, it is clear that the consumed energy decreases with $\beta$ due to the fact that only few files are with high demand when $\beta$ increases.

%%%%%%%%%%%%%%%%%%%%%%%%%%%%%%%%%%%%%%%%%%%%%%%%%%%%%%%%%%%%%%%
\begin{figure}[t]%[htbp]
\centering
\includegraphics[width=0.40\textwidth]{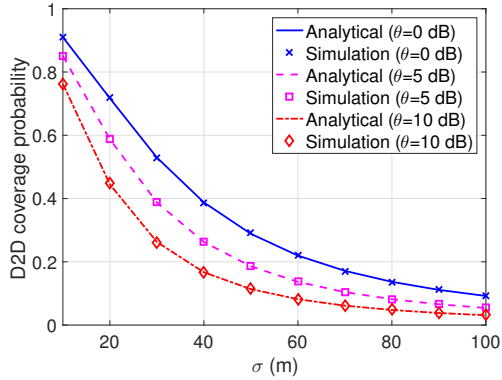}					
\caption {\ac{D2D} coverage probability versus displacement standard deviation $\sigma$. }
\label{verify_energy}
%\vspace{-0.35cm}
\end{figure}
%%%%%%%%%%%%%%%%%%%%%%%%%%%%%%%%%%%%%%%%%%%%%%%%%%%%%%%%%%%%%%%

 %%%%%%% %%%%%%% %%%%%%% %%%%%%% %%%%%%% %%%%%%% %%%%%%%%%%%%%%%%%%%%%%%%%%%%%%%%%%%%%%%%%%%%%%%%%%%%%
 %%%%%%%%%%%%%%%%%%%%%%%%%%%%%%%%%%%%%%%%%%%%%%%%%%%%%%%%%%%%%%%
\begin{figure}[htbp]	%[t]%[htbp]
\centering
\includegraphics[width=0.40\textwidth]{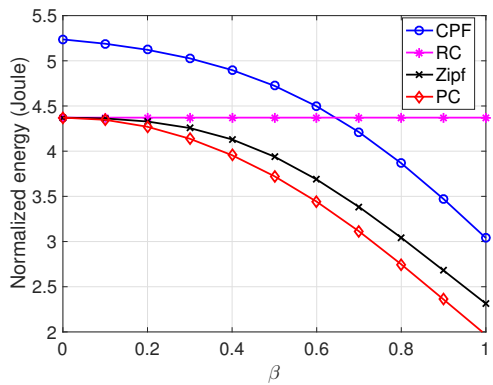}		%energy_vs_beta2		
\caption {Normalized energy consumption versus popularity exponent $\beta$.} %\textbf{Remember Zipf caching}} under different caching schemes
\label{energy_vs_beta}
%\vspace{-0.4cm}
\end{figure}

%%\vspace{-0.4cm}
\section{Conclusion}
%%\vspace{-0.4cm}
In this work, we formulate and solve an optimization problem to minimize the energy consumption of a clustered \ac{D2D} network with random \ac{PC} incorporated at the devices. Using tools from stochastic geometry, we get closed form expressions of the coverage probability and transmission rates, encompassing the underlying physical layer parameters, e.g. mutual interference and \ac{SIR} distribution.  
We show the convexity of the formulated energy minimization problem and obtain a closed form solution for the file placement probabilities.
Results reveal that the proposed \ac{PC} scheme significantly reduce the consumed energy in the network when  compared to conventional methods, e.g., \ac{CPF}, \ac{RC} and Zipf's caching. 
%%%%%%%%%%%%%%%%%%%%%%%%%%%%%%%%%%%%%%%%%%%%%%%%%%%%%%%%%%%%%%%
\begin{figure}[t]%[htbp]
\centering
\includegraphics[width=0.40\textwidth]{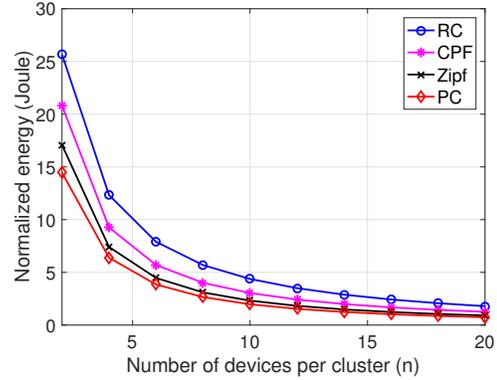}			
\caption {Normalized energy consumption versus the number of devices per cluster.} %\textbf{Remember Zipf caching}}
\label{energy_vs_n}
%\vspace{-0.4cm}
\end{figure}
%%%%%%%%%%%%%%%%%%%%%%%%%%%%%%%%%%%%%%%%%%%%%%%%%%%%%%%%%%%%%%%

\begin{appendices}
\section{Proof of lemma 1}
%\begin{appendix}
Laplace transform of the inter-cluster aggregate interference $I_{\Phi_p^{!}}$ is expressed as
\begin{align}
 &\mathscr{L}_{I_{\Phi_p^{!}}}(s) = \mathbb{E}_{\Phi_p, g_{u}} \Bigg[e^{-s P_d \sum_{\Phi_p \setminus\{x_0\}}  g_{u}  u^{-\alpha}} \Bigg] \nonumber \\
 &\overset{(a)}{=} \mathbb{E}_{\Phi_p} \Bigg[  \mathbb{E}_{g_{u}} \Big[ \prod_{x \in \Phi_p \setminus\{x_0\}}  e^{-s  P_d g_{u}  u^{-\alpha}} \Big] \Bigg] \nonumber \\
&\overset{(b)}{=}  \text{exp}\Bigg(-2\pi\lambda_p \mathbb{E}_{g_{u}} \int_{v=0}^{\infty}\mathbb{E}_{u|v}\Big[1 -
 e^{-s  P_d g_{u}  u^{-\alpha}}  \Big]v\dd{v}\Bigg), %\frac{1}{1 + s p P_d u^{-\alpha}} Rayleigh fading assumption, followed by 
\end{align}
where (a) follows from the independence between devices' locations and the channel gains, (b) follows from the probability generating functional of the \ac{PPP}  \cite{andrews2011tractable}. It is proven in \cite{clustered_twc} that the interfering distance $u$, conditioned on the distance
$v=\lVert x \rVert$ between the interfering cluster center $x \in \Phi_p$ and the typical device is given by $f_U(u|v) = \mathrm{Rice} (u| v, \sigma)$, which represents the Rice's probability density function of parameter $\sigma$. By averaging over $u$, the Laplace transform of  inter-cluster interference can be expressed as
\begin{align}
 &\mathscr{L}_{I_{\Phi_p^{!}}}(s) = \text{exp}\Big(-2\pi\lambda_p \mathbb{E}_{g_{u}} \big[\int_{v=0}^{\infty}\int_{u=0}^{\infty}\big(1 - e^{-s  P_d g_{u}  u^{-\alpha}}  \big) \nonumber \\
 & \quad \quad	\quad\quad \quad \quad\quad\quad\quad \quad\quad\quad\quad\quad\quad \quad	\cdot 
 f_U(u|v)\dd{u}v\dd{v}\big]\Big)   \nonumber \\
 \label{prelaplace}
&\overset{(a)}{=} \text{exp}\Bigg(-2\pi\lambda_p\cdot \nonumber \\
&\Eb_{g_{u}} \underbrace{\int_{v=0}^{\infty}v\dd{v} - \int_{v=0}^{\infty}\int_{u=0}^{\infty}  e^{-s  P_d g_{u}  u^{-\alpha}} f_{U}(u|v)\dd{u} v \dd{v}}_{\mathcal{R}(s,\alpha)}\Bigg) 
 \end{align}
where (a) follows from $\int_{u=0}^{\infty} f_{U}(u|v)\dd{u} =1$. Now, we proceed by calculating the integrands of $\mathcal{R}(s,\alpha)$ as follows.
\begin{align}
\mathcal{R}(s,\alpha)&\overset{(a)}{=} \int_{v=0}^{\infty}v\dd{v} - \int_{u=0}^{\infty}e^{-s  P_d g_{u}  u^{-\alpha}}\cdot \nonumber \\
&\int_{v=0}^{\infty} f_{U}(u|v)v \dd{v}\dd{u}\nonumber \\
&\overset{(b)}{=} \int_{v=0}^{\infty}v\dd{v} - \int_{u=0}^{\infty}e^{-s  P_d g_{u}  u^{-\alpha}}u\dd{u} \nonumber \\
&\overset{(c)}{=}   \int_{u=0}^{\infty}(1 - e^{-s  P_d g_{u}  u^{-\alpha}})u\dd{u} \nonumber \\
&\overset{(d)}{=}   \frac{(s P_d g_{u})^{2/\alpha}}{\alpha} \int_{u=0}^{\infty}(1 - e^{-t})t^{-1-\frac{2}{\alpha}}\dd{u}  \nonumber \\
%&\overset{(d)}{=} \Eb_{g_{u}} \frac{(s P_d g_{u})^{2/\alpha}}{\alpha} \Big(\frac{t^{-2/\alpha}}{-2/\alpha}(1-\exp(-t))\Big|_{0}^{\infty} +\nonumber \\
% &\frac{\alpha}{2} \int_0^\infty \exp(-t)  t^{-2/\alpha} \dd{t}\Big)\nonumber \\
 %&\overset{(d)}{=} \Eb_{g_{u}} \frac{(s P_d g_{u})^{2/\alpha}}{2} \int_0^\infty \exp(-t)  t^{-2/\alpha} \dd{t} \nonumber \\
 %&\overset{(d)}{=} \Eb_{g_{u}} {g_{u}}^{2/\alpha} \frac{(s P_d)^{2/\alpha}}{2}  \Gamma(1 - 2/\alpha)  \nonumber \\
 \label{laplaceppp1}
 &\overset{(e)}{=} \frac{(s P_d)^{2/\alpha}}{2} g_{u}^{2/\alpha} \Gamma(1 + 2/\alpha),	 	
 \end{align}
 where (a) follows from changing the order of integration, (b) follows from $ \int_{v=0}^{\infty} f_{U}(u|v)v\dd{v} = u$, (c) follows from changing the dummy variable $u$ to $v$, (d) follows from changing the variables $t=s g_{u}u^{-\alpha}$, and (e) follows from solving the integration of (d) by parts. Substituting the obtained value for $\mathcal{R}(s,\alpha)$ into (\ref{prelaplace}), and taking the expectation over the exponential random variable $g_u$, with the fact that $\Eb_{g_{u}} g_{u}^{2/\alpha} = \Gamma(1 - 2/\alpha)$, the lemma is proven.
%\begin{align}
% \mathscr{L}_{I_{\Phi_p^{!}}}(s) &= \text{exp}\big(-\pi\lambda_p (s P_d)^{2/\alpha} \Eb_{g_{u}} g_{u}^{2/\alpha} \Gamma(1 + 2/\alpha) \big) 		\nonumber \\
% &\overset{(a)}{=} \text{exp}\big(-\pi\lambda_p (s P_d)^{2/\alpha} \Gamma(1 - 2/\alpha) \Gamma(1 + 2/\alpha) \big)
% \end{align}
%\end{appendix}
%%%%%%%%%%%%%%%%%%%%%%%%%%%%%%%%%%%%%%%%%%%%%%%%%%%%%%%%%%%%%%%
%\begin{appendix}
\section{Proof of lemma 2}
Since the serving and typical devices' locations w.r.t. the cluster center $x_0$ are sampled from normal distribution with variance $\sigma^2$, the serving distance $r$ is Rayleigh distributed with probability density function $f_R(r)=\mathrm{Rayleigh}(r,\sqrt{2}\sigma)$ with a scale parameter $\sqrt{2}\sigma$. Substituting the obtained Laplace transform of (\ref{laplace_trans}) and the probability density function $f_R(r)$ into (\ref{cov_eqn}) yields, %the coverage probability (\ref{v_0}).
 \begin{align} 
 p_{c_d} &=\int_{r=0}^{\infty}
  {\rm e}^{-\pi\lambda_p (sP_d)^{2/\alpha} \Gamma(1 + 2/\alpha)\Gamma(1 - 2/\alpha)}\frac{r}{2\sigma^2}{\rm e}^{\frac{-r^2}{4\sigma^2}}  {\rm dr} , \nonumber \\
  &\overset{(a)}{=}\int_{r=0}^{\infty} \frac{r}{2\sigma^2} 
   {\rm e}^{-\pi\lambda_p \theta^{2/\alpha}r^{2} \Gamma(1 + 2/\alpha)\Gamma(1 - 2/\alpha)}{\rm e}^{\frac{-r^2}{4\sigma^2}}  {\rm dr} , \nonumber \\
%   &=\int_{r=0}^{\infty} \frac{r}{2\sigma^2} 
%   {\rm e}^{-r^2(\pi\lambda_p \theta^{2/\alpha}\Gamma(1 + 2/\alpha)\Gamma(1 - 2/\alpha)+ \frac{1}{4\sigma^2})}  {\rm dr} , \nonumber \\
   &\overset{(b)}{=}\int_{r=0}^{\infty} \frac{r}{2\sigma^2} 
   {\rm e}^{-r^2\mathcal{Z}(\theta,\sigma,\alpha)}  {\rm dr} , \nonumber \\
   &= \frac{1}{4\sigma^2 \mathcal{Z}(\theta,\alpha,\sigma)}
 \end{align}
 where (a) comes from the substitution $s = \frac{\theta r^{\alpha}}{P_d}$, and (b) from $\mathcal{Z}(\theta,\alpha,\sigma)= (\pi\lambda_p \theta^{2/\alpha}\Gamma(1 + 2/\alpha)\Gamma(1 - 2/\alpha)+ \frac{1}{4\sigma^2})$.
 
\end{appendices}

\bibliographystyle{IEEEtran}
\bibliography{bibliography}
\end{document}